\documentclass[11pt]{article}
\usepackage{amsfonts}
\usepackage{mathrsfs}

\usepackage[latin1]{inputenc}
\usepackage{float}
\usepackage{amsmath,amssymb}
\usepackage{latexsym}
\usepackage{epstopdf}
\usepackage{geometry}  
\usepackage[active]{srcltx}
\usepackage{graphicx}
\usepackage{epstopdf}
\usepackage{booktabs}
\usepackage{multirow}
\usepackage{siunitx}
\usepackage{enumitem}

\usepackage[
bookmarks=true,         
bookmarksnumbered=true, 
colorlinks=true, pdfstartview=FitV, linkcolor=blue, citecolor=blue,
urlcolor=blue]{hyperref}

\newtheorem {theorem}{Theorem}[section]

\newtheorem {corollary}{Corollary}[section]

\newtheorem{definition}{Definition}[section]

\newtheorem{lemma}{Lemma}[section]

\newtheorem{remark}{Remark}[section]

\newenvironment{proof}[1][Proof]{\textbf{#1.} }{\
\rule{0.5em}{0.5em}}

\newcommand{\bi}[1]{\mbox{\boldmath{$ #1 $}}}

\def\E{{{\mathbb E}\,}}

\def\Var{{\mathop {{\rm Var\, }}}}

\begin{document}

\title{Jackknife Empirical Likelihood Approach for $K$-sample Tests}
\author{Yongli Sang\textsuperscript{a}\thanks{CONTACT Yongli Sang. Email: yongli.sang@louisiana.edu}, Xin Dang\textsuperscript{b} and Yichuan Zhao\textsuperscript{c}}
\date{%
    \textsuperscript{a}Department of Mathematics, University of Louisiana at Lafayette, Lafayette, LA 70504, USA\\%
    \textsuperscript{b}Department of Mathematics, University of Mississippi, University, MS 38677, USA\\[2ex]%
     \textsuperscript{c}Department of Mathematics and Statistics, Georgia State University, Atlanta, GA 30303, USA\\[2ex]%
    \today
}

\maketitle

\begin{abstract}
\noindent
The categorical Gini correlation is an alternative measure of dependence between a categorical and numerical variables, which characterizes the independence of the variables. 
A  nonparametric test for the equality of $K$ distributions has been developed based on the categorical Gini correlation. By applying the jackknife empirical likelihood approach, the standard limiting chi-square distribution with degree freedom of $K-1$ is established and is used to determine critical value and $p$-value of the test. Simulation studies show that the proposed method is competitive to existing methods in terms of power of the tests in most cases. The proposed method is illustrated in an application on a real data
set.

\noindent  

\vskip.2cm 

\noindent {\bf Keywords:}
\noindent  Energy distance, $K$-sample test, Jackknife empirical likelihood, $U$-statistic, Wilks' theorem, Categorical Gini correlation.

\vskip.2cm 
\noindent  {\textit{MSC 2010 subject classification}: 62G35, 62G20}

\end{abstract}

\section{Introduction}
\noindent
Testing the equality of $K$ distributions from independent random samples is a classical statistical problem encountered in almost every field.  Due to its fundamental importance and wide applications, research for the $K$-sample problem has been kept active since 1940's.  Various tests have been proposed and new tests continue to emerge.  

Often an omnibus test is based on a discrepancy measure among distributions.  For example, the widely used and well-studied tests such as Cram\'{e}r-von Mises test (\cite{Kiefer1959}), Anderson-Darling (\cite{Darling57, Scholz1987}) and their variations utilize different norms on the difference of empirical distribution functions, while some (\cite{Anderson94, Martinez09}) are based on the comparison of density estimators if the underlying distributions are continuous.  Other tests (\cite{Szekely04, Fernandez08}) are based on characteristic function difference measures.  One of such measures is the energy distance (\cite{Szekely13, Szekely17}).  It is the weighted $L_2$ distance between characteristic functions and is defined as follows. 

\begin{definition}[Energy distance]
Suppose that $(\bi X, \bi X') $ and $(\bi Y, \bi Y') $ are independent pairs independently from d-variate distributions $F$ and $G$, respectively. Then the energy distance between $\bi X$ and $\bi Y$ is  
\begin{align} \label{energy-distance}
\mathcal{E}(\bi X,\bi Y)=2\E\|\bi X-\bi Y\|-\E\|\bi X-\bi X'\|-\E\|\bi Y-\bi Y'\|. 
\end{align}
\end{definition}
Let the characteristic functions of $\bi X$ and $\bi Y$ be $\psi_x(\bi t)$ and $\psi_y(\bi t)$, respectively. It has been proved that 
$$\mathcal{E}(\bi X, \bi Y) = c_d \int_{\mathbb{R}^d} \frac{\|\psi_x(\bi t) -\psi_y(\bi t)\|^2}{\| \bi t\|^{d+1}} d\bi t,$$
where $c_d$ is a constant depending on $d$.   Clearly, $\mathcal{E}(\bi X,\bi Y)=0$ if and only if $F=G$. 
A natural estimator of (\ref{energy-distance}), the linear combination of three $U$-statistics,  is called energy statistic.   Reject $F=G$ if the energy statistic is sufficiently large.  To extend to the $K$-sample problem, Rizzo and Sz\'{e}kely (\cite{Rizzo2010}) proposed a new method called distance components (DISCO) by partitioning the total distance dispersion of the pooled samples into the within distance and  between distance components analogous to the variance components in ANOVA. The test statistic is the ratio of the between variation and the within variation, where the between variation is the weighted sum of all two-sample energy distances. Equivalently, Dang {\em et al} \cite{Dang2018} conduced a test based on the ratio of the between variation and the total variation, in which the ratio defines a dependence measure.  Although those tests are consistent against any departure of the null hypothesis and are easy to compute the test statistics, the tests have to reply on a permutation procedure to determine the critical values since the null distribution depends on the unknown underlying distributions. 

Empirical likelihood (EL) tests (\cite{Cao06, EM03,Zhang2007}) successfully avoid the time-consuming permutation procedure.  As a nonparametric approach, the EL (\cite{Owen1988, Owen1990}) also enjoys effectiveness of likelihood method and hence has been widely used, see \cite{Qin2001, Qin1994, Wang1999} and the references therein. We refer to \cite{Chen2008, Chen2015, Cheng2018, Emerson2009} for the updates about the EL in high dimensions.
When the constraints are nonlinear, EL loses this efficiency. To overcome this computational difficulty,  Wood $et$ $al.$ (\cite{Wood1996}) proposed a sequential linearization method by linearizing the nonlinear constraints. However, they did not provide the Wilks' theorem and stated that it was not easy to establish.   Jing {\em et al.} (\cite{Jing2009}) proposed the jackknife empirical likelihood (JEL) approach. The JEL method transforms  the maximization problem of the EL with nonlinear constraints to the simple case of EL on the mean of  jackknife pseudo-values, which is very effective in handling one and two-sample $U$-statistics. This approach has attracted statisticians' strong interest in a wide range of fields due to its efficiency, and many papers are devoted to the investigation of the method.


%

Recently several JEL tests (\cite{Liu2015, Liu2017, Liu2018}) based on characteristic functions have been developed for the two-sample problem.  Wan, Liu and Deng (\cite{Wan2018}) proposed a JEL test using the energy distance, which is a function of three $U$-statistics. To avoid the degenerate problem of $U$ statistics, a nuisance parameter is introduced and the resulting JEL method involves three constraints.  The limiting distribution of the log-likelihood is a weighted chi-squared distribution. Directly generalizing their JEL test to the $K$-sample problem may not work since the number of constraints increases quickly with $K$. There are $K(K+1)/2$ constraints, not only casting difficulty in computation but also bringing challenges in theoretical development. 

We propose a JEL test for the $K$-sample problem with only $K+1$ constraints. We treat the $K$-sample testing problem as a dependence test between a numerical variable and a categorical variable indicating samples from different populations.   We apply JEL with the Gini correlation  that mutually characterizes the dependence (\cite{Dang2018}).  The limiting distribution of the proposed JEL ratio is a standard Chi-squared distribution.  To our best knowledge, our approach is the first consistent JEL test for univariate and multivariate $K$-sample problems in the literature. 
The idea of viewing the $K$-sample test as an independent test between a numerical and categorical variable is not new.  Jiang, Ye and Liu (\cite{Jiang2015}) proposed a nonparametric test based on mutual information. The numerical variable is discretized so that the mutual information can be easily evaluated. However, their method only applies to univariate populations.  Heller, Heller and Gorfine  (\cite{Heller2013, Heller2016}) proposed a dependence test based on rank distances, but their test requires a permutation procedure.

The reminder of the paper is organized as follows. In Section 2, we develop the JEL method for the $K$-sample  test. Simulation studies are conducted in Section 3.  A real data analysis is illustrated in Section 4. Section 5  concludes the paper with a brief summary. All proofs are reserved to the Appendix. 

\section{JEL test for $K$-sample based on a categorical Gini correlation}
Let ${\cal A}_k=\left \{\bi X^{k}_{1},  \bi X^{k}_{2}, ..., \bi X^{k}_{n_k}\right \}$ be a sample from $d$-variate distribution $F_k, \ k=1,...,K,$ respectively.  The pooled sample is denoted as ${\cal A}= {\cal A}_1\cup {\cal A}_2...\cup {\cal A}_K$ of sample size $n = n_1+n_2+...+n_K$. 
The objective is to test the equality of the $K$ distributions, that is, 
\begin{align}\label{test:k-sample}
H_0: F_1=...=F_K \;\;  \text{vs.}  \;\; H_a: F_j \neq F_k \ \ \text{for some $1\leq j < k\leq K$}.
\end{align}

 Let $Z$ be the categorical variable taking values $1, 2, ..., K$, and let $\bi X$ be a continuous random variable in $\mathbb{R}^d$ with  the conditional distribution of $\bi X$ given $Z=k$ being $F_k$. Assume $P(Z=k)=\alpha_k>0$. Then the distribution of $\bi X$ is the mixture distribution $F$  defined as 
 $$F=\sum_{k=1}^K \alpha_k F_k. $$ 
 Treating $\hat\alpha_k = n_k/n$ as an unbiased and consistent estimator of $\alpha_k$, we can view the pooled sample ${\cal A}$ as a sample from $F$. 
 
By introducing the two variables $\bi X$ and $Z$, testing \eqref{test:k-sample} is equivalent to testing the independence between $\bi X$ and  $Z$. We will adopt the recently proposed categorical Gini correlation (\cite{Dang2018}) which characterizes the independence of the continuous and categorical variables.

\subsection{Categorical Gini correlation}
Let $\bi X_1$ and $\bi X_2$ be i.i.d. copies from $F$, and $\bi X_1^k$ and  $\bi X_2^{k}$ be i.i.d. copies from $F_k, \ k=1, ..., K$. 
Let 
\begin{align}\label{gmd}
\Delta=\mathbb{E}\|\bi X_1-\bi X_2\|, \ \ \Delta_k=\mathbb{E}\|\bi X_1^k-\bi X_2^{k}\|, \ k=1,..., K, 
\end{align}
be the Gini distance of $\bi X$ and $\bi X^k$, respectively. 
Then the  Gini correlation (\cite{Dang2018}) between a continuous random variable and a categorical variable is defined as 
\begin{definition}[Dang $et$ $al.$ \cite{Dang2018}]
For a non-degenerate random vector $\bi X$ in $\mathbb{R}^d$ and a categorical variable $Z$, if $\mathbb{E}\|\bi X\|<\infty$, the Gini correlation of $\bi X$ and $Z$ is defined as 
\begin{align}
\rho_g(\bi X, Z)=\dfrac{\Delta-\sum_{k=1}^K \alpha_k \Delta_k}{\Delta}.
\end{align}
\end{definition}
The Gini correlation $\rho_g(\bi X, Z)$ characterizes the dependence. That is, $\rho_g(\bi X,Z)=0$ if and only if $\bi X$ and $Z$ are independent. This is because 
\begin{align*}
&S := \Delta-\sum_{k=1}^K \alpha_k \Delta_k = \sum_{k=1}^K \alpha_k {\cal E} (\bi X^k, \bi X)  \\
&= c_d \sum_{k=1}^K \alpha_k \int \frac{\|\psi_k(\bi t)-\psi(\bi t)\|^2}{\|\bi t\|^{d+1}} d \bi t \geq 0, 
\end{align*}
where $c_d$ is a constant depending on $d$,  $\psi_k(\bi t)$ and $\psi(\bi t)$ are characteristic functions of $\bi X^k$ and $\bi X$, respectively.  
Hence we have the following result,
\begin{lemma} [Dang {\em et al}. \cite{Dang2018}] 
For $\E \|\bi X\| <\infty$, $S=0$ if and only if $F_1=F_2=...=F_K$.
\end{lemma}
Therefore, testing (\ref{test:k-sample}) will be equivalent to testing whether $S=0$. We can rewrite $S$ as
\begin{align*}
S=\sum_{k=1}^K\alpha_k (\E \|\bi X_1-\bi X_2\|-\E \|\bi X_1^k-\bi X_2^{k}\|), 
\end{align*}
which can be estimated unbiasedly by
\begin{align}\label{U:K-sample}
U_{n_1,...,n_K} =\sum_{k=1}^K \hat{\alpha}_k (U_n-U_{n_k}),
\end{align}
where $\hat{\alpha}_k=n_k/n$, 
\begin{align*}
U_n={n \choose 2}^{-1} \sum_{1\leq i <j \leq n}\|\bi X_i-\bi X_j\|,
\end{align*} and
\begin{align*}
U_{n_k}={n_k \choose 2}^{-1} \sum_{1\leq l <m\leq n_k}\|\bi X^{k}_l- \bi X^{k}_m\|,\;\; k=1,...,K.
\end{align*}
Clearly, $U_n$ and $U_{n_k}$ are $U$-statistics of degree 2 with the kernel being $h(\bi x_1,\bi x_2)=\|\bi x_1-\bi x_2\|$. $U_n$ and $U_{n_k}$ are unbiased estimators of $\Delta$ and $\Delta_k$, respectively. 

Under $H_0$, we have $\Delta=\Delta_1=...=\Delta_K$. Conversely, $\Delta=\Delta_1=...=\Delta_K$. Then $S=0$ and hence $F_1=F_2=...=F_K$. Therefore,
Testing $H_0: F_1=...=F_K$ is equivalent to testing 
\begin{align} \label{EUtest}
H^{\prime}_0: \E U_n=\E U_{n_1}=...=\E U_{n_K}.
\end{align}

JEL has been proven to be very effective in dealing with $U$-statistics \cite{Jing2009}, and therefore we will utilize
the JEL approach to test (\ref{EUtest}). 

 \subsection{JEL test for $K$-sample}
In order to apply JEL, we define the 
the corresponding jackknife pseudo-values for $U_n, \ U_{n_k}, k=1,..,K$ as
\begin{align*}
&\hat{V}_i=n U_n-(n-1)U^{(-i)}_{n-1}, \ \text{$i=1, ..., n$}\\
&\hat{V}^{k}_l=n_k U_{n_k}-(n_k-1)U^{(-l)}_{n_k}, \ \text{$l=1, ..., n_k$}
\end{align*}
where
\begin{align*}
U^{(-i)}_{n-1}={n-1 \choose 2}^{-1} \sum_{1\leq j <m \leq n, j,m\neq i}\|\bi X_j-\bi X_m\|,
\end{align*}
and
\begin{align*}
U^{(-l)}_{n_k}={n_k-1 \choose 2}^{-1} \sum_{1\leq j <m \leq n_k, j,m\neq l}\|\bi X^{k}_j- \bi X^{k}_m\|.
\end{align*}
It is obvious to see that 
\begin{align*}
&U_n=\dfrac{1}{n} \sum_{i=1}^{n}\hat{V}_i, \\
&U_{n_k}=\dfrac{1}{n_k} \sum_{l=1}^{n_k}\hat{V}^{k}_l, \;\;\;\mbox{for}\; k = 1,...,K.
\end{align*}
Under $H_0$, we have
\begin{align*}
\E \hat{V}_i=\E \hat{V}^{k}_l=\theta_0,\ i=1, ..., n; \ l=1,...,n_k; \ k=1,...,K, 
\end{align*}
where $\theta_0=\E \|\bi X_1-\bi X_2\|=\E \|\bi X_1^1 -\bi X_2^1\|=...= \E \|\bi X_1^K -\bi X_2^K\|$, with the expectations taking under $H_0$. 

Next, we apply the JEL to the above jackknife pseudo values. 
Let $\bi p_k=(p_{k1},...,p_{kn_k})$ be the empirical probability vector assigned to the elements of ${\cal A}_k$, $k=1,...,K$, and  $\bi p=(p_1,...,p_n)$ be probability vector for $\cal A$.   We have the following optimization problem. 

\begin{equation} \label{Optfun}
R =\max_{\bi p_k,\bi p, \theta} \left \{\left (\prod_{k=1}^{K} \prod_{l=1}^{n_k}n_kp_{kl} \right ) \left ( \prod_{i=1}^{n} n p_i\right)\right \},\\
\end{equation}
subject to the following constraints
\begin{align}
&p_{kl} \ge 0, \ l=1,...,n_k, \ \sum_{l=1}^{n_k} p_{kl}=1, \ 1\leq k \leq  K;  \nonumber\\
&p_i\geq 0,  \  i =1,...,n, \ \sum_{i=1}^{n}p_i=1; \nonumber \\
& \sum_{i=1}^{n}p_i \left ( \hat{V}_i-\theta \right )=0; \ \sum_{l=1}^{n_k}p_{kl} \left(\hat{V}^{k}_l-\theta\right )=0, \ k=1,...,K. \label{constraint}
\end{align}
\begin{remark}
$R$ in equation (\ref{Optfun}) maximizes the squared standard jackknife empirical likelihood ratio (JELR). This is because  $p_i$ is the marginal probability and $p_{kl}$ is the conditional probability and then we have $\prod_{i=1}^n np_i = \prod_{k=1}^K \prod_{l=1}^{n_k} n_k p_{kl}$.  The maximization in $R$ is the same maximization solution of the regular JELR. 
\end{remark}
Applying Lagrange multiplier, one has
\begin{align*}
&p_{kl}=\frac{1}{n_k} \dfrac{1}{1+\lambda_{k} \left(\hat{V}^{k}_l-\theta\right)},  \  l=1, ..., n_k,\ k=1,..,K,\\
&p_i=\frac{1}{n} \dfrac{1}{1+\lambda \left ( \hat{V}_i-\theta \right )}, \ i=1, ..., n,
\end{align*}
where $(\lambda, \lambda_1..., \lambda_{K}, \theta)$ satisfy the following $K+2$ equations: 
\begin{align} \label{findtheta}
&\sum_{i=1}^{n} \frac{\hat{V}_i-\theta}{1+\lambda \left(\hat{V}_i-\theta\right )}=0, \nonumber \\
&\sum_{l=1}^{n_k} \frac{\hat{V}^{k}_l-\theta}{1+\lambda_{k} \left(\hat{V}^{k}_l-\theta\right )}=0, \ k=1,...,K, \nonumber \\
&\lambda\sum_{i=1}^{n} \frac{-1}{1+\lambda \left(\hat{V}_i-\theta\right )}+\sum_{k=1}^K\lambda_{k}\sum_{l=1}^{n_k} \frac{-1}{1+\lambda_{k} \left(\hat{V}^{k}_l-\theta\right )}=0.
\end{align}
In Lemma \ref{Liu5.5}, we have proved the existence of the solutions for the above equations in the Appendix. We denote the solution of (\ref{findtheta}) as $(\tilde{\lambda}, \tilde{\lambda}_1,...,\tilde{\lambda}_K, \tilde{\theta})$. Thus we have the jackknife empirical log-likelihood ratio 
\begin{align}
&-2\log R = -2\sum_{k=1}^K\sum_{l=1}^{n_k} \log(n_kp_{kl}) -2 \sum_{i=1}^n \log(np_i)  \nonumber\\
&= 2\sum_{k=1}^K \sum_{l=1}^{n_k}\log \left \{1+\tilde{\lambda}_{k}  \left(\hat{V}^{k}_l-\tilde{\theta}\right ) \right \} +2 \sum_{i=1}^{n}\log \left \{1+\tilde{\lambda}  \left(\hat{V}_i-\tilde{\theta}\right )
\right \}.
\end{align}
Define $g(\bi x)=\E \|\bi x -\bi X\|, \ \sigma^2_{g}=\Var (g(\bi X));\ 
g_k(\bi x)=\E\|\bi x-\bi X^{k}\|, \ \sigma^2_{g_k}=\Var (g_k(\bi X^k)), \ k=1, ..., K$ and assume
\begin{itemize}
\item \textbf{C1.} $0< \sigma_{g_k}<\infty, \ k=1,...,K$;
\item \textbf{C2.} $\dfrac{n_k}{n} \to \alpha_k>0,\ k=1,...,K$ and $\alpha_1+\alpha_2+...+\alpha_K=1.$
\end{itemize}
Note that \textbf{C1} implies $0 < \sigma_g<\infty$. We have the following Wilks' theorem.
\begin{theorem}\label{wilk-ksample}
Under $H_0$ and the conditions \textbf{C1} and \textbf{C2}, we have 
$$ -2 \log  R \stackrel{d}{\rightarrow} \chi^{2}_{K-1}, \;\;\;\text{as $n \to \infty$}.$$
\end{theorem}
\begin{proof}
See the Appendix. 
\end{proof}

As a special case of the $K$-sample test, the following result holds for $K=2$.
\begin{corollary}\label{wilk-2sample}
For the two-sample problem, under the conditions \textbf{C1-C2} and $H_0$, we have 
$$ -2 \log  R\stackrel{d}{\rightarrow} \chi^{2}_1, \;\;\;\text{as $n \to \infty$}.$$
\end{corollary}

\begin{remark}
Compared with the result of \cite{Wan2018}, the limiting distribution of the proposed empirical log-likelihood ratio is a standard chi-squared distribution. The empirical log-likelihood has no need for multiplying a factor to adjust unbalanced sample sizes. 
\end{remark}

\begin{remark}
Our JEL approach considers energy distance of $\bi X^k$ and $\bi X$, while the JEL method in Wan {\em et al}. \cite{Wan2018} utilizes energy distance of between classes $\bi X^k$ and $\bi X^j$. For $K\geq 3$, they need to deal with $K(K+1)/2$ constraints, the number much larger than $K+1$ of ours in (\ref{constraint}). 
\end{remark}

With Therorm \ref{wilk-ksample}, we reject $H_0$ if the observed jackknife empirical likelihood $-2\log \hat{R} $ is greater than $\chi^2_{K-1}(1-\alpha)$, where $\chi^2_{K-1}(1-\alpha)$ is the $100(1-\alpha)\%$ quantile of $\chi^2$ distribution with $K-1$ degrees of freedom. The $p$-value of the test can be calculated by
 $$ p\mbox{-value} = P_{H_0}(\chi^2_{K-1} > -2 \log \hat{R} ), $$
 and the power of the test is 
$$ \mbox{power} = P_{H_a} (-2 \log R > \chi^2_{K-1}(1-\alpha)). $$
In the next theorem, we establish the consistence of the proposed test, which states  its power is tending to 1 as the sample size goes to infinity. 
\begin{theorem}\label{consis:test}
Under the conditions \textbf{C1} and\textbf{C2}, the proposed JEL test for the K-sample problem is consistent for any fixed alternative. That is,
$$ P_{H_a} (-2 \log R > \chi^2_{K-1}(1-\alpha)) \rightarrow 1, \;\;\;\text{as $n \to \infty$}. $$
\end{theorem}
\begin{proof}
See the Appendix. 
\end{proof}

\section{Simulation Study}
In order to assess the proposed JEL method for the homogeneity testing, we conduct extensive simulation studies in this section. We compare the following methods.
\begin{description}
\item[JEL-S:]our proposed JEL method.  R package ``dfoptim" \cite{Varadhan18} is used for solving the equation system of (\ref{findtheta}). 
\item[JEL-W:] the JEL approach proposed in \cite{Wan2018}. It is applied only for $K=2$. 
\item[ET:] the DISCO test of \cite{Rizzo2010}. Its null limiting distribution of the test statistic depends on the underlying distribution and hence the test is implemented by the permutation procedure. Function ``eqdist.etest" with the default number of replicates in  R package ``energy" is used \cite{Rizzo17}. 
\item[AD:] the Anderson-Darling test of \cite{Anderson52, Scholz1987}. The procedure ``ad.test" in R package ``kSamples" is used \cite{Scholz2018}.  
\item[KW:] the Kruskal-Wallis test of \cite{Kruskal1952, Puri71} implemented in R package ``kSamples". 
\item[HHG:] the HHG test of \cite{Heller2013, Heller2016}. The test is performed by a permutation procedure that is implemented in R package ``HHG" \cite{Brill2017}. 
\end{description}  

Type I error rates and powers for each method at significance levels $\alpha=0.05$ and $\alpha=0.10$ are based on 10,000 replications. The results at significance level $\alpha=0.10$ are similar to the results at 0.05 level and hence are not presented. We only consider one case of $K=2$ to demonstrate the similarity of our JEL-S and JEL-W. The remaining cases are for $K=3$ without loss of  generality. We generate univariate ($d=1$) and multivariate ($d=3$, $d=6$) random samples from normal, heavy-tailed $t$ and asymmetric exponential distributions. In each distribution, samples of balanced and unbalanced sample sizes are generated.   

\subsection{Normal distributions}

We first compare our JEL-S with JEL-W, which is also a JEL approach based on energy statistics but designed for the two-sample problem.  We generate two independent samples with either equal  ($n_1=n_2=50$) or unequal sample sizes ($n_1=40, n_2=60$) from the $d$-dimensional normal distributions $N(\bi 0_d, \bi I_{d\times d})$ and $N(\bi 0_d, \delta_1 \bi I_{d\times d})$, respectively, where $\bi 0_d$ is the $d$-dimensional zero vector, $\bi I_{d\times d}$ is the identity matrix in $d$ dimension and $\delta_1$ is a positive number to specify the difference of scales. The results are displayed in Table \ref{tab-k=2}.

\begin{table}[H]
\centering
\caption{Type I error rates and powers of tests for normal distributions with different scales for $K=2$ under the significance level $\alpha=0.05$.  \label{tab-k=2}}
\vspace*{0.1in}
\setlength\tabcolsep{2.5pt}
\small
\begin{tabular}{llcccccccccccr}\hline\hline
           &   &   \multicolumn{6}{c}{$n_1=n_2=50$}&
\multicolumn{6}{c}{$n_1=40$, $n_2=60$}\\  
 $d$           & $\delta_1$  & JEL-S &JEL-W &ET & AD & KW & HHG  &  \  \ \ \ \ JEL-S &JEL-W& ET & AD & KW&HHG  \\   \hline



1&1                           &.053&.052&.046&.052&.051&.050 & \  \ \ \ \ .048&.046&.044&.048&.048&.051\\
&1.5                       &.753&.773&.224&.250&.053&.542 & \  \ \ \ \ .733&.747&.172&.194&.042&.489\\
&2.0                         &.986&.988&.755&.762&.060&.961 & \  \ \ \ \ .984&.983&.667&.685&.044&.955\\

&2.5                       &.992&.996&.970&.969&.067&1.00  & \  \ \ \ \  .993&.995&.950&.950&.045&.999\\
&3.0                          &.9961&.995&.998&.997&.069&1.00 & \  \ \ \ \  .995&.997&.996&.995&.045&1.00\\ \hline

3&1                           &.045&.048&.048&.050&.049&.049 & \  \ \ \ \ .048&.047&.043&.048&.049&.049\\
&1.5                         &.654&.653&.104&.211&.051&.551 & \  \ \ \ \ .618&.626&.082&.176&.047&.508\\
&2.0                          &.967&.964&.347&.673&.051&.963  & \  \ \ \ \ .965&.964&.252&.609&.041&.946\\
&2.5                       &.988&.988&.697&.944&.055&.999 & \  \ \ \ \ .990&.992&.582&.917&.043&.998\\
&3.0                        &.993&.990&.911&.994&.054&1.00  & \  \ \ \ \ .993&.991&.836&.990&.041&1.00\\ \hline


6&1                           &.047&.048&.041&.048&.048&.047 & \  \ \ \ \ .043&.043&.043&.048&.049&.047\\

&1.5                          &.910&.911&.145&.472&.050&.896 & \  \ \ \ \ .891&.894&.107&.426&.045&.875\\

&2.0                         &.990&.989&.612&.977&.056&1.00 & \  \ \ \ \ .993&.990&.487&.970&.041&1.00\\
&2.5                           &.996&.996&.948&1.00&.060&1.00  & \  \ \ \ \ .995&.994&.893&1.00&.040&1.00\\
&3.0                          &.996&.995&.998&1.00&.058&1.00  & \  \ \ \ \ .996&.995&.994&1.00&.040&1.00\\ \hline\hline
\end{tabular}
\end{table}

As expected, the JEL-W  and our approach perform similarly because both are JEL approach on energy distance to compare two samples.   Advantages of the JEL approach over the others in testing scale differences are the same for $K=3$, which is demonstrated in the following simulation. 
 
Three random samples ${\cal A}_1$, ${\cal A}_2$ and ${\cal A}_3$ are simulated  from normal distributions of $N(\bi 0, \bi I_{d \times d})$,  $N(\bi 0, \delta_2 \bi I_{d\times d})$ and  $N(\bi 0,  \delta_3\bi I_{d \times d})$ respectively, where $\delta_2$ and $\delta_3$ are positive numbers.  The simulation result is shown  in Table \ref{tab-scale}.

\begin{table}[H]
\begin{center}
\begin{small}
\caption{sizes and powers of tests for normal distributions with different scales for $K=3$ under the significance level $\alpha=0.05$.  \label{tab-scale}}
\vspace*{0.1in}
\setlength\tabcolsep{4.5pt}
\begin{tabular}{llcccccccccr}\hline\hline
               & &   \multicolumn{5}{c}{$ n_1=n_2=n_3=50$}&
\multicolumn{5}{c}{$n_1=40,n_2=60,n_3=50$}\\ 
           
$d$&$(\delta_2, \delta_3)$ & JEL-S &ET & AD & KW & HHG & \  \ \ \ \  JEL-S & ET & AD & KW&HHG  \\ \hline

1&(1, 1)                        &.058&.046&.053&.050&.052   & \  \ \ \ \ .057&.044&.050&.048&.050\\
&(1.1, 1.5)                     &.749&.181&.206&.052&.478  & \  \ \ \ \ .718&.164&.190&.051&.400\\
&(1.2, 2.0)                    &.988&.659&.680&.055&.943 & \  \ \ \ \ .985&.625&.638&.057&.892\\
&(1.3, 2.5)                    &.995&.936&.936&.059&.998& \  \ \ \ \ .993&.916&.914&.056&.991\\
&(1.4, 3.0)                    &.995&.994&.992&.062&.999& \  \ \ \ \ .994&.989&.985&.054&.999\\\hline


3&(1, 1)                   &.057&.045&.051&.051&.050 & \  \ \ \ \ .053&.043&.052&.048&.051\\
&(1.1, 1.5)                  &.622&.093&.164&.051&.474.   & \  \ \ \ \ .599&.088&.149&.049&.398\\
&(1.2, 2.0)             &.975&.268&.569&.050&.936  & \  \ \ \ \ .969&.245&.529&.052&.878\\
&(1.3, 2.5)                    &.998&.548&.877&.052&.996 & \  \ \ \ \ .998&.508&.847&.054&.988\\
&(1.4, 3.0)                    &.999&.808&.978&.057&.999   & \  \ \ \ \ .999&.756&.965&.054&.999\\\hline


6&(1, 1)                        &.053&.045&.055&.056&.048 & \  \ \ \ \ .055&.046&.050&.050&.049\\
&(1.1, 1.5)                  &.906&.124&.369&.049&.852  & \  \ \ \ \ .885&.114&.337&.049&.774\\
&(1.2, 2.0)                  &.999&.465&.942&.052&.999 & \  \ \ \ \ 1.00&.416&.924&.048&.997\\
&(1.3, 2.5)                    &1.00&.857&1.00&.055&1.00 & \  \ \ \ \ 1.00&.818&.998&.049&1.00\\
&(1.4, 3.0)                    &1.00&.985&1.00&.055&1.00 & \  \ \ \ \ .999&.971&1.00&.050&1.00\\\hline\hline

\end{tabular}
\end{small}
\end{center}
\end{table}

In Table \ref{tab-scale}, the size of tests are given in the rows of $(1,1)$ and the powers in other rows.   We can see that every method maintains the nominal level well.  As expected, KW performs badly for scale differences because KW is a nonparametric one-way ANOVA on ranks and it is inconsistent for scale-difference problem.  Although ET and AD are consistent, they are less powerful than the JEL method and HHG.  The JEL method always has the highest power among the all considered tests.

Next, we consider the location difference case. Three random samples ${\cal A}_1$, ${\cal A}_2$ and ${\cal A}_3$ are simulated  from normal distributions of $N(\bi 0, \bi I_{d \times d})$, $N(\delta_4 \bi 1_{d},   \bi I_{d \times d})$ and $N( \delta_5 \bi 1_{d}, \bi I_{d \times d})$, respectively. Here $\bi 1_{d}$ is the $d$-vector with all elements being 1.  The sizes of the tests are reported in the rows of $(0,0)$ in Table \ref{tab-location} and the others rows provide the powers of the tests. 

\begin{table}[H]
\begin{center}
\begin{small}
\caption{Type I error rates and powers of tests for normal distributions with different locations under the significance level $\alpha=0.05$.\label{tab-location}}
\vspace*{0.1in}
\setlength\tabcolsep{4.5pt}
\begin{tabular}{llcccccccccr}\hline\hline
               & &   \multicolumn{5}{c}{$n_1=n_2=n_3=50$}&
\multicolumn{5}{c}{$n_1=40$, $n_2=60$, $n_3=50$}\\ 
$d$ &$(\delta_4, \delta_5)$ & JEL-S &ET & AD & KW & HHG & \  \ \ \ \  JEL-S & ET & AD & KW&HHG  \\ \hline

1&(0, 0)      &.057 &.042  &.046 &.045  &.046               & \  \ \ \ \ .065&.043  &.049&.047&.051 \\
&(0.2, 0.4) &.060 &.362 &.386  &.395 &.252               & \  \ \ \ \ .063&.333  &.347&.348 &.170\\
&(0.4, 0.8) &.074  &.930 &.940 &.943 &.822                & \  \ \ \ \ .068&.907  &.911&.913&.648\\
&(0.6, 1.2)  &.175 &1.00 &1.00  &1.00 &.994               & \  \ \ \ \ .141&.999 &.999&.999 &.971\\
&(0.8, 1.6)  &.569 &1.00  &1.00 &1.00 &1.00                & \  \ \ \ \ .457&1.00  &1.00&1.00&1.00 \\
 \hline


3&(0, 0)      &.056 &.049  &.052&.052  &.050&. \  \ \ \ \ 048&.042  &.052&.052&.049 \\
&(0.2, 0.4) & .059 &.690  &.857&.867&.504& \  \ \ \ \ .050&.645  &.818&.823 &.319\\
&(0.4, 0.8) &.105&1.00 &1.00 &1.00&.998& \  \ \ \ \ .085&1.00  &1.00&1.00&.975\\
&(0.6, 1.2)  &.607&1.00 &1.00  &1.00 &1.00& \  \ \ \ \ .432&1.00  &1.00&1.00 &1.00\\
&(0.8, 1.6)  & .996&1.00&1.00 &1.00 &1.00& \  \ \ \ \ .932&1.00  &1.00&1.00&1.00 \\
 \hline
 

6&(0, 0)      & .056&.046  &.050&.049  &.052& \  \ \ \ \ .056&.047  &.048&.049&.055 \\
&(0.2, 0.4) & .055&.915 &.995  &.996&.680& \  \ \ \ \ .060&.878  &.983&.985 &.452\\
&(0.4, 0.8) &.189&1.00 &1.00 &1.00&1.00& \  \ \ \ \ .143&1.00  &1.00&1.00&.999\\
&(0.6, 1.2)  &.985&1.00 &1.00  &1.00 &1.00& \  \ \ \ \ .926&1.00  &1.00&1.00 &1.00\\
&(0.8, 1.6)  & .999&1.00&1.00 &1.00 &1.00& \  \ \ \ \ .999&1.00  &1.00&1.00&1.00 \\
 \hline\hline

\end{tabular}
\end{small}
\end{center}
\end{table}

The Type I error rates of all tests are close to the nominal level. The JEL-S performs the worst with the lowest power in this case, although it is consistent for any alternatives. An intuitive interpretation is that the JEL assigns more weights on the sample points lying between classes and loses power to differentiate classes. The phenomenon of less power in the location-difference problem is also common for the density approach, as mentioned in \cite{Martinez09}.   For the location difference problem, we suggest to use non-parametric tests based on distribution function approaches.  For example,  AD and KW tests are recommended.  

Our JEL-S has low powers to test location differences, it, however, is sensitive to detect scale-location changes.   Three random samples ${\cal A}_1$, ${\cal A}_2$ and ${\cal A}_3$ are simulated  from normal distributions $N(\bi 0, \bi I_{d \times d})$, $N(\delta_6 \bi 1_{d},  \delta_8 \bi I_{d \times d})$ and $N( \delta_7 \bi 1_{d}, \delta_9\bi I_{d \times d})$,  respectively. Here $\delta_i, i=6,...,9$ measure the difference of locations and scales. The simulation results are reported in Table \ref{tab-scale-location}.

\begin{table}[H]
\begin{center}
\begin{small}
\caption{Type I error rates and powers of tests for normal distributions with different locations and scales under the significance level $\alpha=0.05$. \label{tab-scale-location}}
\setlength\tabcolsep{3.5pt}
\begin{tabular}{llcccccccccr}\hline\hline
               & &   \multicolumn{5}{c}{$n_1=n_2=n_3=50$}&
\multicolumn{5}{c}{$n_1=40, n_2=60, n_3=50$}\\ 

$d$&$(\delta_6, \delta_7, \delta_8, \delta_9)$ & JEL-S &ET & AD & KW & HHG & \  \ \ \ \  JEL-S & ET & AD & KW&HHG  \\ \hline


1&(0, 0, 1, 1)                           &.061&.047 &.051& .051& .051  & \  \ \ \ \ .061&.044 &.049& .048& .049 \\
&(0.1, 0.2, 1.2, 1.4)             &.540&.173  &.195&.102 &.335 & \  \ \ \ \ .497&.153  &.163&.091 &.208\\
&(0.2, 0.4, 1.4, 1.8)             &.952&.570  &.603 &.208&.841 & \  \ \ \ \ .933&.504  &.510 &.185&.630\\
&(0.3, 0.6, 1.6, 2.2)                 &.990&.885  &.893&.336 &.987 & \  \ \ \ \ .989&.827  &.819&.290 &.905\\
&(0.4, 0.8, 1.8, 2.6)               &.988&.981  &.981&.434&.999   & \  \ \ \ \ .990&.964  &.958&.384&.986 \\
 \hline


3&(0, 0, 1, 1)                          &.054&.041 &.048& .049& .046 & \  \ \ \ \ .057&.046 &.051& .051& .050 \\
&(0.1, 0.2, 1.2, 1.4)               &.436&.204  &.348&.260&.405  & \  \ \ \ \ .392&.174  &.300&.230&.240\\
&(0.2, 0.4, 1.4, 1.8)               &.897&.702  &.891 &.723&.926 & \  \ \ \ \ .852&.650  &.849 &.683&.762\\
&(0.3, 0.6, 1.6, 2.2)                 &.993&.970  &.997&.958 &.998 & \  \ \ \ \ .985&.947  &.993&.939 &.976\\
&(0.4, 0.8, 1.8, 2.6)             &.999&.999  &1.00&.996&1.00 & \  \ \ \ \ .997&.997  &1.00&.995&.999 \\
 \hline


6&(0, 0, 1, 1)                           &.054&.048 &.053& .052& .051 & \  \ \ \ \ .051&.044 &.051& .050& .047 \\
&(0.1, 0.2, 1.2, 1.4)               &.729&.313 &.670&.478&.729  & \  \ \ \ \ .668&.265 &.600&.434&.509\\
&(0.2, 0.4, 1.4, 1.8)                &.995&.930  &.999 &.963&.999 & \  \ \ \ \ .990&.886  &.996 &.943&.982\\
&(0.3, 0.6, 1.6, 2.2)                &.999&1.00  &1.00&1.00 & 1.00  &\  \ \ \ \ .999&.999  &1.00&.999 &1.00\\
&(0.4, 0.8, 1.8, 2.6)               &.998&1.00 &1.00&1.00&1.00   &\  \ \ \ \ .999&1.00  &1.00&1.00&1.00 \\
 \hline\hline

\end{tabular}
\end{small}
\end{center}
\end{table}

From Table \ref{tab-scale-location}, we can have the following observations. For $d=1$, KW is the least powerful.  ET and KW perform similar but worse than HHG and JEL-S.  JEL-S has the highest powers. For example, JEL-S is about 20\%-30\% more powerful than the second best HHG method in the case of $(\delta_6, \delta_7, \delta_8, \delta_9)=(0.1,0.2,1.2,1.4)$.   For $d=3$ and $d=6$, ET performs the worst and JEL-S is the most competitive method. 

\subsection{Heavy-tailed distribution: $t(5)$}
We compare the performance of JEL-S with others in the heavy-tailed distributions. Three random samples ${\cal A}_1$, ${\cal A}_2$ and ${\cal A}_3$ are simulated  from multivariate $t$ distributions with 5 degrees of freedom with the same locations $\bi 0$ and different scales $\bi I_{d \times d}, \delta_{10} \bi I_{d \times d}$ and $\delta_{11}\bi I_{d \times d}$, respectively.  The results are reported in Table \ref{tab-t5}.

\begin{table}[H]
\begin{center}
\begin{small}
\caption{Type I error rates and powers of tests for $t_5$ distributions with difference scales under the level $\alpha=0.05$. \label{tab-t5}}
\vspace*{0.05in}
\setlength\tabcolsep{4.5pt}
\begin{tabular}{llcccccccccr}\hline \hline
               & &   \multicolumn{5}{c}{$ n_1=n_2=n_3=50$}&
\multicolumn{5}{c}{$n_1=40$, $n_2=60$, $n_3=50$}\\ 

$d$ &$(\delta_{10}, \delta_{11})$ & JEL-S &ET & AD & KW & HHG & \  \ \ \ \  JEL-S & ET & AD & KW&HHG  \\ \hline


1&(1, 1)                              &.079&.041&.047&.045&.045 & \  \ \ \ \ .079&.043&.051&.047&.052\\
&(1.1, 1.5)                       &.206&.063&.068&.048&.120 & \  \ \ \ \ .196&.060&.067&.051&.104\\
&(1.2, 2.0)                       &.435&.117&.116&.053&.283 & \  \ \ \ \ .405&.104&.105&.047&.229\\
&(1.3, 2.5)                      &.630&.199&.194&.055&.478 & \  \ \ \ \ .615&.189&.181&.052&.389\\
&(1.4, 3.0)                      &.769&.310&.287&.055&.647  & \  \ \ \ \ .750&.271&.250&.051&.530\\
\hline


3&(1, 1)                             &.078&.046&.061&.049&.049 & \  \ \ \ \ .076&.044&.060&.051&.046\\
&(1.1, 1.5)                       &.355&.093&.148&.050&.303 & \  \ \ \ \ .337&.080&.144&.049&.249\\
&(1.2, 2.0)                        &.733&.228&.410&.051&.730 & \  \ \ \ \ .710&.202&.386&.048&.636\\
&(1.3, 2.5)                     &.912&.463&.707&.050&.938 & \  \ \ \ \ .898&.416&.669&.049&.876\\
&(1.4, 3.0)                      &.959&.677&.882&.051&.988 & \  \ \ \ \ .957&.626&.854&.054&.963\\
\hline

6&(1, 1)                             &.076&.040&.079&.050&.052   & \  \ \ \ \ .075&.046&.081&.052&.051\\

&(1.1, 1.5)                 &.442&.121&.332&.049&.457 & \  \ \ \ \ .435&.113&.323&.048&.391\\

&(1.2, 2.0)                      &.844&.396&.795&.056&.912 & \  \ \ \ \ .837&.354&.766&.052&.851\\

&(1.3, 2.5)                     &.953&.711&.965&.046&.993 & \  \ \ \ \ .951&.658&.951&.051&.980\\

&(1.4, 3.0)                   &.977&.900&.996&.051&1.00 & \  \ \ \ \ .976&.866&.992&.050&.998\\
\hline\hline

\end{tabular}
\end{small}
\end{center}
\end{table}

Compared with results of the normal distribution case in Table \ref{tab-scale}, the power of every method in Table \ref{tab-t5} has been impacted by heavy-tailed outliers, while impacts in high dimensions are less than that in one dimension.  JEL-S has a slight over-size problem. Its size is 2-3\% higher than the nominal level, while its power is uniformly the highest among all methods. For the small difference case with $(\delta_8,\delta_{9})=(1,1.5)$, JEL-S is 10\% more powerful than the second best HHG method. 

\subsection{Non-symmetric distribution: Exponential distribution}
Lastly we consider the performance of JEL-S for asymmetric distributions. We generate random samples ${\cal A}_1$, ${\cal A}_2$ and ${\cal A}_3$ from multi-variate exponential distributions with independent components.  The components of each sample are simulated from  exp(1), exp$(\delta_{12})$ and exp$(\delta_{13})$, respectively. Type I error rates  and powers are presented in Table \ref{tab-exp}.

\begin{table}[H]
\begin{center}
\begin{small}
\caption{Type I error rates and powers under exponential distributions with different scales.}\label{tab-exp}
\vspace*{0.05in}
\setlength\tabcolsep{4.5pt}
\begin{tabular}{lccccccccccr}\hline\hline
               & &   \multicolumn{5}{c}{$n_1=n_2=n_3=50$}&
\multicolumn{5}{c}{$n_1=40, n_2=60, n_3=50$}\\ 

$d$&$(\delta_{12}, \delta_{13})$ & JEL-S &ET & AD & KW & HHG &  \  \ \ \ \ JEL-S & ET & AD & KW&HHG  \\ \hline


1&(1, 1)                              &.089&.044&.047&.048&.048 & \  \ \ \ \ .090&.046&.053&.053&.050\\
&(1.1, 1.2)                        &.140&.098&.098&.095&.091 & \  \ \ \ \ .143&.091&.094&.093&.074\\
&(1.2, 1.4)                        &.276&.245&.236&.231&.205 & \  \ \ \ \ .268&.222&.209&.204&.132\\
&(1.3, 1.6)                        &.458&.447&.423&.407&.365 & \  \ \ \ \ .427&.417&.378&.371&.236\\
&(1.4, 2.2)                        &.866&.906&.878&.851&.837 & \  \ \ \ \ .830&.885&.844&.813&.699\\
\hline


3 &(1, 1)                              &.073&.041&.047&.047&.047  & \  \ \ \ \ .074&.049&.052&.049&.051\\
&(1.1, 1.2)                       &.227&.166&.215&.209&.192& \  \ \ \ \ .211&.153&.195&.192&.128\\
&(1.2, 1.4)                       &.563&.544&.626&.603&.571& \  \ \ \ \ .529&.492&.571&.551&.386\\
&(1.3, 1.6)                        &.835&.871&.917&.899&.873  & \  \ \ \ \ .799&.827&.873&.854&.707\\
&(1.4, 2.2)                        &.989&.999&1.00&.999&.999 & \  \ \ \ \ .987&.999&.999&.999&.998\\
\hline


6&(1, 1)                              &.068&.047&.049&.047&.048   & \  \ \ \ \ .072&.045&.052&.051&.052\\
&(1.1, 1.2)                       &.351&.269&.395&.385&.339 & \  \ \ \ \ .332&.245&.373&.361&.217\\
&(1.2, 1.4)                        &.815&.806&.913&.893&.846 & \  \ \ \ \ .772&.763&.886&.866&.682\\
&(1.3, 1.6)                        &.972&.990&.997&.996&.992  & \  \ \ \ \ .954&.977&.996&.993&.949\\
&(1.4, 2.2)                        &.991&1.00&1.00&1.00&1.00 & \  \ \ \ \ .990&1.00&1.00&1.00&1.00\\
\hline\hline

\end{tabular}
\end{small}
\end{center}
\end{table}

From Table \ref{tab-exp}, we observe that JEL-S  suffers slightly from the over-size problem, while the problem becomes less of an issue for higher dimensions. JEL-S performs the best when the differences are small. HHG is inferior to others. Asymmetric exponential distributions with different scales also imply different mean values, and hence KW performs fairly.

\subsection{Summary of the simulation study}
Some conclusions can be drawn across all tables 1-6.  HHG is affected by unbalanced sizes the most among all methods. For example, in Table \ref{tab-scale-location}, the power of HHG is dropped 13\% and 17\% for $d=1$ and $d=3$, respectively,  from the equal size  to the unequal size case, compared with a 3-5\% decrease in other methods.  

Considering the same total size,  the power in balanced sample is higher than unequal size samples for all tests.  All methods share the same pattern of power changes when the dimension changes. For the Normal scale difference cases,  powers in $d=3$ are lower than those in $d=1$ and $d=6$. While for $t(5)$ and exponential distributions, powers increase with $d$.  

Overall, JEL-S is competitive to the current approaches for comparing $K$-samples. Particularly, JEL-S  is very powerful for the scale difference problems and is very sensitive to detect subtle differences among distributions. 


\section{Real data analysis}
For the  illustration purpose, we apply the proposed JEL approach to a multiple two-sample test example. 
We apply the JEL method to the banknote authentication data which is available in UCI Machine Learning Repository (\cite{Lohweg2013}).  The data set consists of 1372 samples with 762 samples of them  from the Genuine class denoted as Gdata and 610 from the Forgery class denoted as Fdata. 
Four features are recorded  from each sample: variance of wavelet transformed image (VW), skewness of wavelet transformed image (SW), kurtosis of wavelet transformed image (KW) and entropy of image (EI). 
One can refer to Lohweg $et$ $al.$ (\cite{Lohweg2013}) and Sang, Dang and Zhao (\cite{Sang2019}) for more descriptions and information of the data.

The densities of each of the variables for each class are drawn in Figure 1. We observe that the distributions of each variable in different classes are quite different, especially for variables VW and SW. The locations of VW in two classes are clearly different. The distribution of SW shows some multimodal trends in both classes. The distribution of KW in Forgery class is more right-skewed than it is in Genuine class. EI of two classes has similar left-skewed distribution. Here we shall compare the multivariate distribution of two classes and also conduct univariate two-class tests on each of four variables.

%
%
%
%
%

\begin{figure}[H]
\centering
\label{densityfig}
\begin{tabular}{cc}
\includegraphics[width=2.3in,height=2.3in]{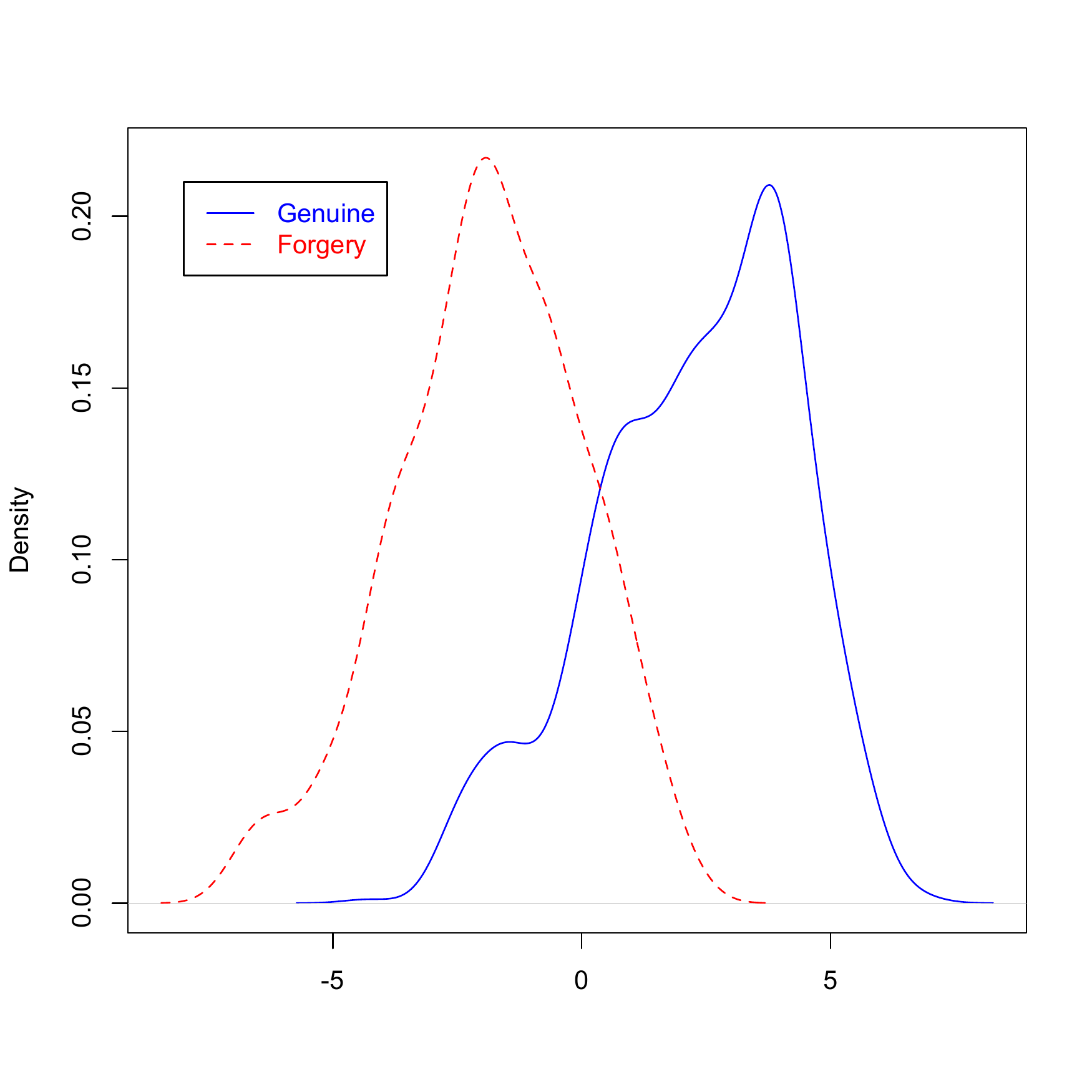} \vspace{-0.1in}&
\includegraphics[width=2.3in,height=2.3in]{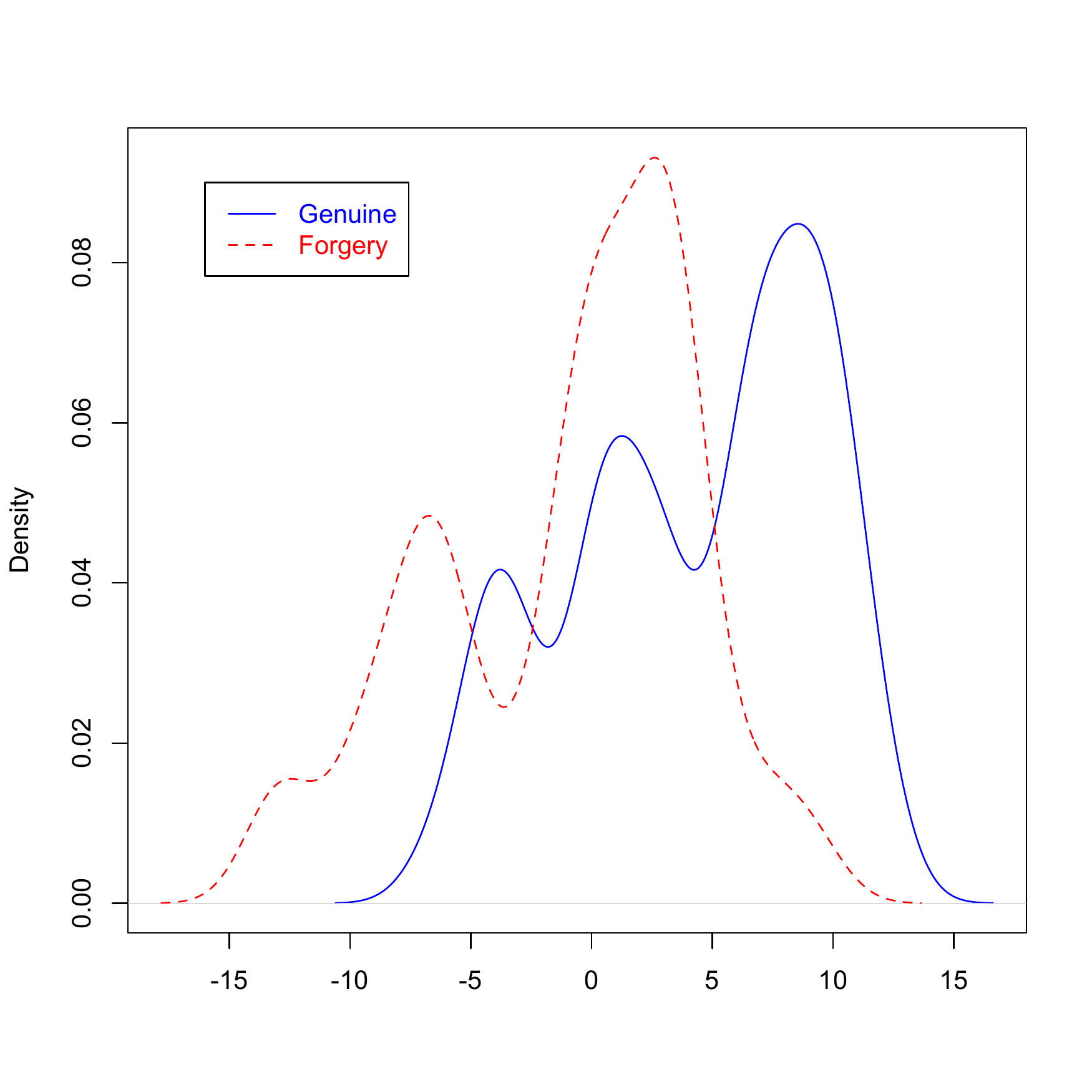} \vspace{-0.1in}\\
(a) VW\vspace{-0.1in}&(b)  SW\vspace{-0.1in}\\
\includegraphics[width=2.3in,height=2.3in]{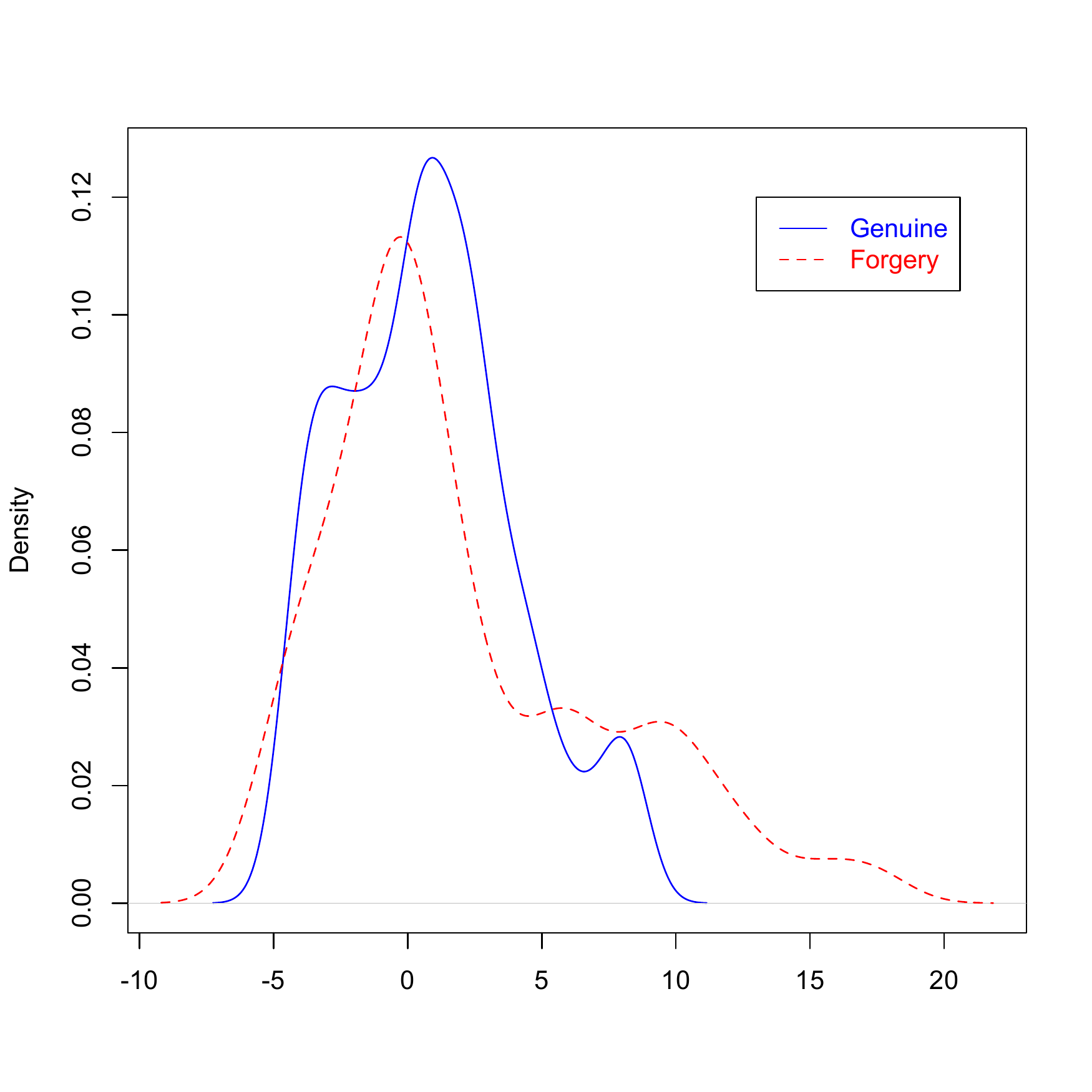} \vspace{-0.1in}&
\includegraphics[width=2.3in,height=2.3in]{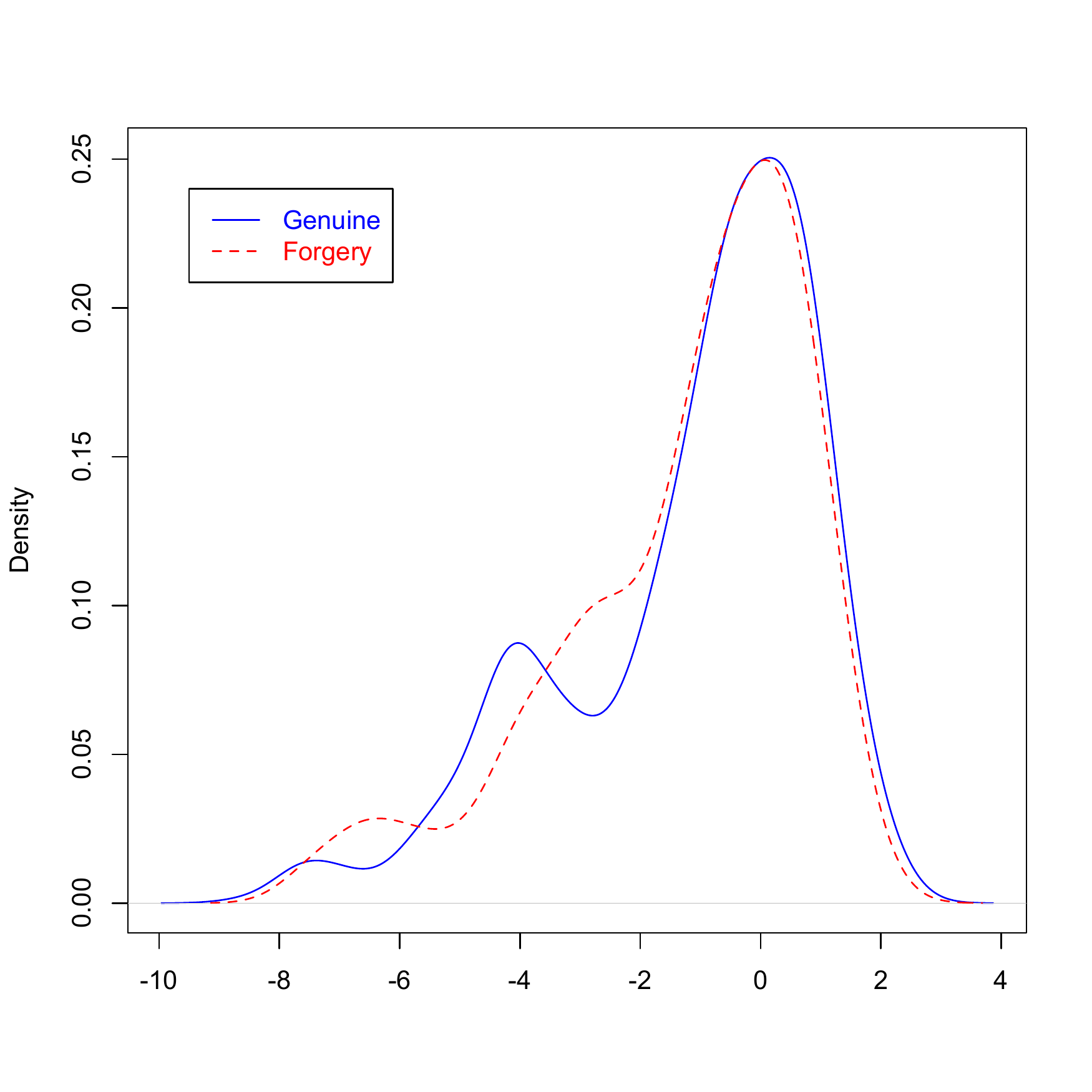} \vspace{-0.1in}\\
(c) KW&(d)  EI\\
\end{tabular}
\caption{Density of each of variables.}
\label{fig:IF}
\end{figure}

From Table \ref{tab-bn}, all tests reject the equality of multivariate distributions of Gdata and Fdata with significantly small $p$-values close to 0.  Also the $p$-values for testing separately the individual distributions of VW, SW and KW are small for all methods and thus we conclude that the underlying distributions of those variables are quite different in two classes. For EI variable, however, we do not have significant evidence to reject the equality of the underlying distributions. This result agrees well with the impression from the last graph (d) in Figure \ref{densityfig}.  In these tests, the $p$-values calculated from JEL approaches are much higher than those calculated from ET, AD, KW and HHG.  As expected, our method performs very similar to the JEL-W approach for the two-sample problem.

\begin{table}[H]
\begin{center}
\begin{small}
\caption{$P$-values of tests }\label{tab-bn}
\vspace*{0.1in}
\begin{tabular}{lcccccr}\hline\hline          
 Variable(s)&  JEL-S &JEL-W & ET & AD & KW&HHG  \\ \hline
(Gdata, Fdata)             &0&0&.005&5.4e-112& 0&.0001 \\
(Gdata-VW, Fdata-VW)        &0&0&.0050&2.6e-205& 0  &.0010  \\
(Gdata-SW, Fdata-SW)        &0&0&.0050&1.3e-74& 0&.0010       \\
(Gdata-KW, Fdata-KW)         &0&0&.0050&4.47e-10& .0226  &.0010  \\
(Gdata-EI, Fdata-EI)          &.4748&.4759&.1950&.1168& .2253  &.1389  \\\hline\hline
\end{tabular}
\end{small}
\end{center}
\end{table}

\section{Conclusion}

In this paper, we have extended the JEL method to the $K$-sample test via the categorical Gini correlation. Standard limiting chi-square distributions with $K-1$ degrees of  freedom are established and are used to conduct hypothesis testings without a permutation procedure. Numerical studies confirm the advantages of the proposed method under a variety of situations.  One of important contributions of this paper is to develop a powerful nonparametric method for multivariate $K$-sample problem. 

Although the proposed $K$-sample JEL test is much more sensitive to shape difference among $K$ distributions, it is dull to detect the variation in location when the differences are subtle. 
This disadvantage probably stems from finding the solution of $\theta$ in equations of (\ref{findtheta}). That is, the within Gini distances and  the overall Gini distances are restricted to be the same. This forces the JEL approach weighing more on the observations that are more close to other distributions. As a result, the JEL approach loses some power to detect the difference among the locations.  This is a common problem for tests based on density functions. For the location difference problem, distribution function approaches such as AD and KW are more preferred. 

Furthermore, the proposed JEL approach is developed based on Euclidean distance, and hence is only invariant under translation and homogeneous changes. 
Dang $et$ $al.$ (\cite{Dang2018}) suggested an affine Gini correlation, and we will continue this work by proposing an affine JEL test. 


\section{Appendix}
Define  $\bi \lambda=(\lambda, \lambda_{1}, ..., \lambda_{K})$,
\begin{align*}
&W_{0n}(\theta, \bi \lambda)=\dfrac{1}{n}\sum_{i=1}^{n} \frac{\hat{V}_i-\theta}{1+\lambda \left(\hat{V}_i-\theta\right )},\\
&W_{kn}(\theta, \bi \lambda)=\dfrac{1}{n} \sum_{l=1}^{n_k} \frac{\hat{V}^{k}_l-\theta}{1+\lambda_{k} \left(\hat{V}^{k}_l-\theta\right )}, \  k=1,...,K,\\
&W_{{(K+1)n}}(\theta, \bi \lambda)=\dfrac{1}{n}\sum_{k=1}^K\lambda_{k}\sum_{l=1}^{n_k} \frac{-1}{1+\lambda_{k} \left(\hat{V}^{k}_l-\theta\right )}+\dfrac{\lambda}{n}\sum_{i=1}^{n} \dfrac{-1}{1+\lambda \left (\hat{V}_{i}-\theta \right) }.
\end{align*}

\begin{lemma}[Hoeffding, 1948]\label{Hoeffding} 
Under condition \textbf{C1}, 
 $$\dfrac{\sqrt{n_k}(U_{n_k}-\theta_0)}{2\sigma_{gk}} \stackrel{d}{\rightarrow} N(0, 1) \ \text{ as}  \ n_k \to \infty,$$
 $$\dfrac{\sqrt{n}(U_{n}-\theta_0)}{2\sigma_{g}} \stackrel{d}{\rightarrow} N(0, 1) \ \text{ as}  \ n \to \infty.$$
\end{lemma}

\begin{lemma}
Let $S_{k}=\dfrac{1}{n_k} \displaystyle\sum_{l=1}^{n_k} \left (\hat{V}^{k}_l- \theta_0\right )^2,  k=1, ..., K,$ and 
$S_0=\dfrac{1}{n} \displaystyle\sum_{i=1}^{n}\left ( \hat{V}_i- \theta_0\right )^2.$  
Under the conditions of Lemma \ref{Hoeffding}, 
\begin{itemize}
\item (i) $S_{k}=4\sigma^2_{g_k}+o_p(1)$, as $n_k \to \infty, k=1, ..., K$
\item (ii) $S_0=4\sigma^2_g+o_p(1)$, as $n \to \infty$.
\end{itemize}
\end{lemma} 

\begin{lemma}[Liu, Liu and Zhou, 2018]\label{Liu5.5}
Under conditions \textbf{C1} and \textbf{C2} and $H_0$, with probability tending to one as $\min\{n_1, ..., n_K\} \to \infty$, there exists a root $\tilde{\theta}$ of 
\begin{align*}
&W_{kn}(\theta, \bi \lambda)=0, \ k=0, 1,..., K+1,
\end{align*}
such that $|\tilde{\theta}-\theta_0|< \delta,$ where $\delta=n^{-1/3}$. 
\end{lemma}

Let $\tilde{\bi \eta}=(\tilde{\theta}, \tilde{\bi \lambda})^T$ be the solution to the above equations, and $\bi \eta_0=(\theta, 0, ..., 0)^T$. By expanding $W_{kn}(\tilde{\bi \eta})$ at $\bi \eta_0$, we have, for $k=0, 1,...,K+1$,
\begin{align*}
&0=W_{kn}(\bi \eta_0)+ \dfrac{\partial W_{kn}}{\partial \theta}(\bi \eta_0)(\tilde{\theta}-\theta_0)+\dfrac{\partial W_{kn}}{\partial \lambda}(\bi \eta_0) \tilde{\lambda}\\
&+\dfrac{\partial W_{kn}}{\partial \lambda_1}(\bi \eta_0) \tilde{\lambda}_{1}+...+\dfrac{\partial W_{kn}}{\partial \lambda_K}(\bi \eta_0) \tilde{\lambda}_{K}+R_{kn},
\end{align*}
where $R_{kn}=\frac{1}{2} (\tilde{\bi \eta}-\bi \eta_0)^T \dfrac{\partial^2 W_{kn}(\bi \eta^{*})}{\partial \bi \eta \partial \bi \eta^T} (\tilde{\bi \eta}-\bi \eta_0)=o_p(n^{-1/2}),$ and $\bi \eta^{*}$ lies between $\bi \eta_0$ and $ \tilde{\bi \eta}$.

\begin{lemma}\label{cov:u}
Under $H_0$, $\mbox{Cov}(U_{n_k}, U_{n_l})=0$, $1\leq k \neq l \leq K$; $\mbox{Cov}(U_n, U_{n_k})=4/n \sigma^2_{g_k}+O(1/n^2)$, $k=1,..,K$.
\end{lemma}
\begin{proof}
First of all, $k$ sample and $l$ sample are independent and hence $\mbox{Cov}(U_{n_k}, U_{n_l})=0$ for $1\leq k \neq l \leq K$. For the covariance between $U_n$ and $U_{n_k}$, we have the Hoeffding decompositions
\begin{align*}
&U_{n_k}=\theta_0+\frac{2}{n_k}\sum_{l=1}^{n_k}h_1(X^{k}_l)+\frac{2}{n_k (n_k-1)}\sum_{1\leq l<m \leq n_k}h_2(X^{k}_l, X^{k}_m), \\
&U_n=\theta_0+\frac{2}{n}\sum_{i=1}^{n}h_1(X_i)+\frac{2}{n (n-1)}\sum_{1\leq l<m\leq n}h_2(X_l, X_m),
\end{align*}
where $h_1(x)=\E h(x, X)-\theta_0$  and  $h_2(x, y)=h(x, y)-h_1(x)-h_1(y)+\theta_0$.  Then
\begin{eqnarray*}
 Cov(U_n, U_{n_k}) &=& \frac{2}{n_k}\frac{2}{n} n_k \E h_1^2(X_1^k)+ \frac{2}{n_k(n_k-1)} \frac{2}{n(n-1)}n_k(n_k-1) \E h_2^2(X_1^k, X_2^k)\\  
& = & \frac{4}{n} \sigma_{gk}^2+ O(1/n^2).
\end{eqnarray*} 
This completes the proof of Lemma 6.4.
\end{proof}

\noindent\textbf{Proof of Theorem \ref{wilk-ksample}.}
$$
\begin{pmatrix}
W_{0n}(\bi \eta_0)\\
   W_{1n}(\bi \eta_0)   \\
  W_{2n}(\bi \eta_0)  \\
    \vdots \\
  W_{Kn}(\bi \eta_0))\\
 0\\
\end{pmatrix}
= \bi{\mathcal{B}}
\begin{pmatrix}
\tilde{\lambda}\\
    \tilde{\lambda}_{1}\\
\tilde{\lambda}_{2} \\
    \vdots \\
\tilde{\lambda}_{K}\\
   \tilde{\theta}-\theta_0
    \end{pmatrix}+o_{p}(n^{-1/2}),
$$
where 
$$  \bi{\mathcal{B}}=
\begin{pmatrix}
  \sigma^2&0&0&\hdots&0&0& 1 \\
 0&\alpha_1 \sigma^2_1&0&\hdots&0&0&\alpha_1\\
    \hdotsfor{7} \\
    0&0&0&\hdots&0&\alpha_K \sigma^2_K&\alpha_K\\
      1&\alpha_1&\alpha_2&\hdots&\alpha_{K-1}&\alpha_{K}&0\\   
\end{pmatrix}_{(K+2)\times (K+2)},$$
$\sigma^2=4\sigma^2_{g}$ and $\sigma^2_k=4\sigma^2_{gk}, k=1,..,K.$ 

It is easy to see that $\bi{\mathcal{B}}$ is nonsingular under Conditions {\bf C1} and {\bf C2}. Therefore,
$$\begin{pmatrix}
\tilde{\lambda}\\
    \tilde{\lambda}_{1}\\
\tilde{\lambda}_{2} \\
    \vdots \\
\tilde{\lambda}_{K}\\
   \tilde{\theta}-\theta_0
    \end{pmatrix}=\bi{\mathcal{B}}^{-1}\begin{pmatrix}
   W_{0n}(\bi \eta_0)   \\
  W_{1n}(\bi \eta_0)  \\
    \vdots \\
  W_{Kn}(\bi \eta_0))\\
  0\\
\end{pmatrix}+o_{p}(n^{-1/2}).
$$

Under $H_0$, $\sigma^2_1=...=\sigma^2_K= \sigma^2$.
\begin{align*}
&\tilde{\theta}-\theta_0=\sum^K_{k=0} \dfrac{1}{2} W_{kn}(\bi \eta_0)+o_p(n^{-1/2}).
\end{align*}
By \cite{Jing2009}, we have 
\begin{align*}
&\tilde{\lambda}=\dfrac{U_{n}-\tilde{\theta}}{\tilde{S}_{0}}+o_{p}(n^{-1/2})\; \mbox{ and }\;\tilde{\lambda}_{k}=\dfrac{U_{n_k}-\tilde{\theta}}{\tilde{S}_{k}}+o_{p}(n^{-1/2}), \ k=1,...,K,
\end{align*}
where 
$\tilde{S}_{0}=\dfrac{1}{n}\displaystyle\sum^{n}_{i=1}(\hat{V}_{i}-\tilde{\theta})^2$ and $\tilde{S}_{k}=\dfrac{1}{n_k}\displaystyle\sum^{n_k}_{l=1}(\hat{V}^{k}_{l}-\tilde{\theta})^2$. It is easy to check that $\tilde{S}_{0}=\sigma^2+o_p(1)$ and $\tilde{S}_{k}=\sigma^2_k+o_p(1),\ k=0,1,..,K$.
By the proof of Theorem 1 in \cite{Jing2009},
\begin{align*}
-2\log R=\left [ \sum_{k=1}^K \dfrac{n_{k}(U_{n_k}-\tilde{\theta})^2}{\sigma^2_{k}}+\frac{n(U_n-\tilde{\theta})^2}{\sigma^2} \right](1+o_{p}(1)).
\end{align*}
With simple algebra, we have 
\begin{align}\label{decomposition}
&\frac{n(U_n-\tilde{\theta})^2}{\sigma^2}+\sum_{k=1}^K \dfrac{n_{k}(U_{n_k}-\tilde{\theta})^2}{\sigma^2_{k}}\nonumber \\
&=(\sqrt{n}W_{0n}(\bi \eta_0), ..., \sqrt{n}W_{Kn}(\bi \eta_0) \times \bi{A}^T \bi{\mathcal{W}} \bi{A} \times(\sqrt{n}W_{0n}(\bi \eta_0), ..., \sqrt{n}W_{Kn}(\bi \eta_0))^T +o_p(1),\nonumber\\
\end{align}
where
$$ \bi{A}=
\begin{pmatrix}
1/2&-1/2&-1/2&\hdots&-1/2&-1/2&\hdots&-1/2\\
 -1/2&\alpha^{-1}_{1}-1/2&-1/2&\hdots&-1/2&-1/2&\hdots&-1/2\\
    \hdotsfor{8} \\
   -1/2&-1/2&-1/2&\hdots&-1/2&-1/2&\hdots&\alpha^{-1}_{K}-1/2\\
\end{pmatrix}$$
and
$$ \bi{\mathcal{W}}=
\begin{pmatrix}
1/\sigma^2&0&0&\hdots&0\\
 0&\alpha_1/\sigma^2_1&0&\hdots&0\\
    \hdotsfor{5} \\
    0& 0&0&\hdots&\alpha_K/\sigma^2_K\\
\end{pmatrix}.$$
Furthermore, by \cite{Shi1984}, we have the central limit theorem for $W$'s at $\bi \eta_0$. This is because each $W(\bi \eta_0)$ is the average of asymptotically independent psuedo-values. That is,  
$$
\sqrt{n}\begin{pmatrix}
W_{0n}(\bi \eta_0))\\
   W_{1n}(\bi \eta_0)   \\
  W_{2n}(\bi \eta_0)  \\
    \vdots \\
  W_{Kn}(\bi \eta_0))\\
\end{pmatrix}
\overset{D}{\to} N(\bi 0, \bi \Sigma),
$$
where 

$$ \bi \Sigma=
\begin{pmatrix}
\sigma^2&\alpha_1 \sigma^2_1&\alpha_2 \sigma^2_2&\hdots&\alpha_K \sigma^2_K\\
 \alpha_1 \sigma^2_1&\alpha_1 \sigma^2_1&0&\hdots&0\\
\alpha_2 \sigma^2_2&0&\alpha_2 \sigma^2_2&\hdots&0\\
    \hdotsfor{5} \\
    \alpha_K \sigma^2_K& 0&0&\hdots&\alpha_K \sigma^2_K\\
\end{pmatrix}.$$

Therefore, under $H_0$, $-2\log R$ converges to $\sum_{i=1}^{K+1} \omega_i \chi^2_i$ in distribution, where $\chi^2_i, i=1,..., K+1$ are $K+1$ independent chi-square random variables with one degree of freedom, and $\omega_i, i=1,.., K+1$ are eigenvalues of 
  $\bi \Sigma_0^{1/2} \bi A^T \bi{\mathcal{W}_0}  \bi {A} \bi \Sigma_0^{1/2}$, where 
$$ \bi \Sigma_0=
\begin{pmatrix}
1&\alpha_1&\alpha_2&\hdots&\alpha_K\\
 \alpha_1&\alpha_1&0&\hdots&0\\
    \hdotsfor{5} \\
    \alpha_K& 0&0&\hdots&\alpha_K\\
\end{pmatrix}$$
and 
$$ \bi{\mathcal{W}_0}=
\begin{pmatrix}
1&0&0&\hdots&0\\
 0&\alpha_1&0&\hdots&0\\
    \hdotsfor{5} \\
    0& 0&0&\hdots&\alpha_K\\
\end{pmatrix}.$$


We can show that $\bi {A}^{T} \bi{\mathcal{W}_0} \bi{A}=\bi{A}$. Hence, $\bi \Sigma_0^{1/2} \bi {A}^T \bi{\mathcal{W}_0}  \bi {A} \bi \Sigma_0^{1/2}=\bi \Sigma_0 \bi{A}$ since $\bi{A}$ is symmetric. With algebra calculation,  the eigenvalues of $\bi \Sigma_0 \bi{A}$ are \{0, 0, 1,...,1\} with trace($\bi \Sigma_0 \bi{A})=K-1.$ By this result, we  complete the proof.    $\hfill{\Box}$
\vspace{0.1in}

\noindent \textbf{Proof of Theorem \ref{consis:test}.}
Under $H_a$, at least one of $\mathbb{E}U_k,$ $k=1,...,K$ will be different from the others. Let $\mathbb{E}U_k=\theta_k, \ k=1,...,K.$
From (\ref{decomposition}), 
\begin{eqnarray*}
-2\log R&=&\dfrac{n(U_n-\tilde{\theta})^2}{\tilde{S}_0^2}+\sum^K_{k=1}\dfrac{n_k(U_{n_k}-\tilde{\theta})^2}{\tilde{S}^2_i}+o(1)\\
&=&\dfrac{n(U_n-\tilde{\theta})^2}{\tilde{S}_0^2}+ \sum^K_{k=1}\left [\dfrac{\sqrt{n_k}(U_{n_k}-\theta_k)}{\tilde{S}_k}+\dfrac{\sqrt{n_k}(\theta_k-\tilde{\theta})}{\tilde{S}_k} \right ]^2+o(1),
\end{eqnarray*}
which is divergent since at least one of $\dfrac{\sqrt{n_k}(\theta_k-\tilde{\theta})^2}{\tilde{S}_k}, k=1,...,K$ will diverge to $\infty$. $\hfill{\Box}$

%
%



\begin{thebibliography}{99}

\bibitem{Anderson52}
Anderson, T. W., and Darling, D. A. (1952). Asymptotic theory of certain
goodness-of-fit criteria based on stochastic processes. {\em Ann. Math.
Statis.} {\bf 23}, 193-212.

\bibitem{Anderson94}
Anderson, N.H., Hall, P. and Titterington, D.M. (1994). Two-sample test statistics for measuring discrepancies between two multivariate probability density functions using kernel-based density estimates. {\em J. Multivariate Anal.} {\bf 50}, 41-54. 



\bibitem{Brill2017}
Brill, B. and Kaufman, S.  based in part on an earlier implementation by Ruth Heller and Yair Heller. (2017). HHG: Heller-Heller-Gorfine Tests of Independence and Equality of Distributions. R package version 2.2. https://CRAN.R-project.org/package=HHG. 



\bibitem{Cao06}
Cao, R. and Van Keilegom, I. (2006). Empirical likelihood tests for two-sample problems via nonparametric density estimation. {\em Can. J. Stat.},  {\bf 34}(1), 61-77.


\bibitem{Chen2015}
Chen, B.B., Pan, G.M., Yang, Q. and Zhou, W. (2015). Large dimensional empirical likelihood. {\em Statist. Sinica}, {\bf 25}, 1659-1677.

\bibitem{Chen2008}
Chen, J., Variyath, A. M. and Abraham, B.  (2008). Adjusted empirical likelihood and its properties. {\em J. Comput.  Graph.  Statist.} {\bf 17} (2), 426-443.  

\bibitem{Cheng2018}
Cheng, C.H., Liu, Y., Liu, Z. and Zhou, W. (2018). Balanced augmented jackknife empirical likelihood for two sample $U$-statistics. {\em Sci. China Math.}, {\bf 61}, 1129-1138.

\bibitem{Dang2018}
Dang, X., Nguyen, D., Chen, X. and Zhang, J. (2019).  A new Gini correlation between quantitative and qualitative variables. arXiv:1809.09793.

\bibitem{Darling57}
Darling, D.A. (1957). The Kolomogorov-Smirnov, Cram\'{e}r-von Mises tests. {\em Ann. Math. Stat.}, {\bf 28}(4), 823-838.

\bibitem{EM03}
Einmahl, J. and McKeague, I. (2003). Empirical likelihood based hypothesis testing. {\em Bernoulli}, {\bf 9}(2), 267-290.


\bibitem{Emerson2009}
Emerson S. and Owen, A. (2009). Calibration of the empirical likelihood method for a vector mean. {\em Electron. J. Statist.}, {\bf 3}, 1161-1192.


\bibitem{Fernandez08} 
Fern\'{a}ndez, V., Jim\`{e}nez Gamerro, M. and Mu\~{n}oz Garc\`{i}a, J. (2008). A test for the two-sample problem based on empirical characteristic functions. {\em Comput. Stat. Data An.},  {\bf 52}, 3730-3748.

\bibitem{Heller2013}
Heller, R., Heller, Y. and Gorfine, M.  (2013). A consistent multivariate test of association based on ranks of distances. {\em Biometrika} {\bf 100}(2), 503-510.

\bibitem{Heller2016}
Heller, R., Heller, Y., Kaofman, S., Brill, B. and Gorfine, M.  (2016). Consistent distribution-free $K$-sample and independence test for univariate random variables. {\em J. Mach. Learn. Res.} {\bf 17}(29), 1-54.

\bibitem{Jiang2015}
Jiang, B., Ye, C. and Liu, J.  (2015). Nonparametric $K$-sample tests via dynamic slicing. {\em J. Amer. Statist. Assoc.} {\bf 110}, 642-653.

\bibitem{Jing2009}
Jing, B., Yuan, J. and Zhou, W.  (2009). Jackknife empirical likelihood. {\em J. Amer. Statist. Assoc.} {\bf 104}, 1224-1232.


\bibitem{Kiefer1959}
Kiefer, J.  (1959).  $k$-sample analogues of the Kolmogorov-Simirnov, Cram$\acute{e}$r-von Mises tests. {\em Ann. Math. Statist.} {\bf 30}, 420-447.


\bibitem{Kruskal1952}
Kruskal, W.H. and Wallis, W.A.  (1952).  Use of ranks in one-criterion analysis of variance. {\em J. Amer. Statist. Assoc.} {\bf 47}, 583-621.

\bibitem{Liu2017}
Liu, X., Wang, Q., and Liu, Y.  (2017). A consistent jackknife empirical likelihood test for distribution functions. {\em Ann. Inst. Statist. Math.} {\bf 69}(2), 249-269.


\bibitem{Liu2018}
Liu, Y., Liu, Z. and Zhou, W. (2018).  A test for equality of two sample distributions via integrating characteristic functions. {\em Statist. Sinica}, DOI: 10.5705/ss.202017.0236.


\bibitem{Liu2015}
Liu, Z., Xia, X. and Zhou, W. (2015). A test for equality of two distributions via jackknife empirical likelihood and characteristic functions. {\em Comput. Statist. Data Anal.} {\bf 92}, 97-114.

\bibitem{Lohweg2013}
Lohweg, V., Hoffmann, J. L., D$\ddot{o}$rksen, H., Hildebrand, R., Gillich, E., Hofmann, J. and Schaede, J. (2013). Banknote Authentication with Mobile Devices. {\em Proc. SPIE 8665, Media Watermarking, Security, and Forensics.} 


\bibitem{Martinez09}
Mart\`{i}nez-Camblor, P. and de U\~{n}a-\'{A}lvarez, J. (2009). Non-parametric $k$-sample tests: Density functions vs distribution functions. {\em Comput. Statist. Data Anal.} {\bf 53}, 3344-3357. 

\bibitem{Owen1988}
Owen, A. (1988).  Empirical likelihood ratio confidence intervals  for  single functional. {\em Biometrika  }  {\bf 75},   237-249.
%
\bibitem{Owen1990}
Owen, A.  (1990).  Empirical likelihood ratio confidence regions. {\em Ann. Statist.}  {\bf 18},   90-120.

\bibitem{Puri71}
Puri, M. L. and Sen, P. K. (1971). {\em Nonparametric Methods in Multivariate Analysis}, John Wiley and Sons, New York.


\bibitem{Qin2001}
Qin, G.S. and  Jing, B.Y. (2001). Empirical likelihood for censored linear regression. {\em Scand. J. Statist.} {\bf 28}, 661-673.


\bibitem{Qin1994}
Qin, J. and  Lawless, J. (1994). Empirical likelihood and general estimating functions. {\em Ann. Statist.} {\bf 22}, 300-325.


\bibitem{Rizzo2010}
Rizzo, M.L. and  Sz$\acute{e}$kely, G.J. (2010). Disco Analysis: A nonparametric extension of analysis of variance. {\em Ann. Appl. Stat.} {\bf 4}, 1034-1055.

\bibitem{Rizzo17}
Rizzo, M.  and Sz$\acute{e}$kely, G.J.  (2017). E-statistics: multivariate inference
  via the energy of data. R package version 1.7-2. https://CRAN.R-project.org/package=energy




\bibitem{Sang2019}
Sang, Y., Dang, X. and Zhao, Y. (2019) Jackknife empirical likelihood methods for Gini correlation and their equality testing. {\em J. Statist. Plann. Inference}, {\bf 199}, 45-59. 

\bibitem{Scholz1987}
Scholz, F. W.  and Stpephens, M.A.  (1987). $K$-sample Anderson-Darling tests. {\em J. Amer. Statist. Assoc.} {\bf 82}(399), 918-924.

\bibitem{Scholz2018} 
Scholz, F. W. and Zhu, A.  (2018).  K-sample rank tests and their combinations. R package version 1.2-8. https://CRAN.R-project.org/package=kSamples.

\bibitem{Shi1984}
Shi, X. (1984). The approximate independence of jackknife pseudo-values
and the bootstrap Methods. {\em J. Wuhan Inst. Hydra-Electric Engineering}, {\bf 2}, 83-90.
 

\bibitem{Szekely04}
Sz{\'e}kely, G.J. and Rizzo, M.L.  (2004). Testing for equal distributions in high dimension, {\em InterStat}, Nov. (5).  

\bibitem{Szekely13}
Sz{\'e}kely, G.J. and Rizzo, M.L. (2013). Energy statistics: A class of statistics based on distances. {\em J. Stat. Plan. Infer.}, {\bf 143}, 1249-1272.

\bibitem{Szekely17}
Sz{\'e}kely, G.J. and Rizzo, M.L. (2017). The energy of data, {\em Ann. Rev. Stat. Appl.}, {\bf 4} (1), 447-479.

\bibitem{Varadhan18}
Varadhan R., Johns Hopkins University, Borchers, H.W. and ABB Corporate Research (2018). dfoptim: Derivative-Free Optimization. R package version 2018.2-1.
 https://CRAN.R-project.org/package=dfoptim. 



\bibitem{Wan2018}
Wan,Y., Liu, Z. and Deng, M.  (2018). Empirical likelihood test for equality of two distributions using distance of characteristic functions. {\em Statistics}, DOI: 10.1080/02331888.2018.1520855.


\bibitem{Wang1999}
Wang, Q.H. and  Jing, B.Y. (1999). Empirical likelihood for partial linear models with fixed designs. {\em Statist. Probab. Lett.} {\bf 41}, 425-433.

\bibitem{Wood1996}
Wood, A.T.A., Do, K.A., Broom, N.M. (1996). Sequential linearization of empirical likelihood constraints with application to $U$-statistics. {\em J. Comput.  Graph.  Statist.} {\bf 5}, 365-385.


\bibitem{Zhang2007}
Zhang, J. and Wu, Y. (2007). $K$-sample tests based on the likelihood ratio. {\em Comput. Statist. Data Anal.} {\bf 51}, 4682-4691.







 
\end{thebibliography}
\end{document}